\documentclass[a4paper,fleqn,11pt]{article}

\usepackage{amsmath}
\usepackage{amsthm}
\usepackage{amssymb}
\usepackage{accents}

\usepackage[a4paper,top=3cm, bottom=3cm, left=3cm, right=3cm]{geometry}
\usepackage[shortlabels,inline]{enumitem}
\usepackage[pdftex,ocgcolorlinks,pagebackref=false]{hyperref}
\usepackage[affil-it]{authblk}

\setlist[enumerate,1]{label={(\roman*)}}

\usepackage[capitalise,noabbrev]{cleveref}

\Crefname{property}{Property}{Properties}

\theoremstyle{plain}
\newtheorem{theorem}{Theorem}[section]
\newtheorem{lemma}[theorem]{Lemma}
\newtheorem{proposition}[theorem]{Proposition}
\newtheorem{corollary}[theorem]{Corollary}
\newtheorem{remark}[theorem]{Remark}
\newtheorem{example}[theorem]{Example}

\theoremstyle{definition}

\newtheorem{definition}[theorem]{Definition}

\newcommand{\setbuild}[2]{\left\{#1\middle|#2\right\}}

\newcommand{\complexes}{\mathbb{C}}
\newcommand{\nonnegativereals}{\mathbb{R}_{\ge 0}}
\newcommand{\positiveintegers}{\mathbb{N}_{>0}}

\DeclareMathOperator{\tensorrank}{R}
\DeclareMathOperator{\borderrank}{\underline{R}}
\DeclareMathOperator{\asymptoticrank}{\undertilde{R}}

\newcommand{\distributions}[1][]{\mathcal{P}_{#1}}
\DeclareMathOperator{\support}{supp}
\DeclareMathOperator{\probability}{Pr}

\newcommand{\typeclass}[2]{T^{#1}_{#2}}
\newcommand{\entropy}{H}

\newcommand{\unittensor}[1]{\langle{#1}\rangle}
\newcommand{\mmtensor}[3]{\langle{#1,#2,#3}\rangle}
\newcommand{\partialmmtensor}[1]{M_{#1}}
\newcommand{\partialmmsupport}[1]{S_{#1}}
\newcommand{\supportrank}{\tensorrank_{s}}
\newcommand{\bordersupportrank}{\borderrank_{s}}
\newcommand{\asymptoticsupportrank}{\asymptoticrank_{s}}

\title{Partial and weighted matrix multiplication}
\author[1]{P\'eter Vrana}
\affil[1]{Department of Algebra and Geometry, Institute of Mathematics, Budapest University of Technology and Economics, M\H uegyetem rkp. 3., H-1111 Budapest, Hungary.}
\date{}

\begin{document}
\maketitle

\begin{abstract}
In a paper published in 1981, Sch\"onhage showed that large total matrix multiplications can be reduced to powers of partial matrix multiplication tensors, which correspond to the bilinear computation task of multiplying matrices with some of the entries fixed to be zero. It was left as an open problem to generalize the method to the case when the multiplication is also partial in the sense that only a subset of the entries need to be computed. We prove a variant of a more general case: reducing large weighted matrix multiplications to tensor powers of a partial matrix multiplication in the sense that every entry of the result is a partial version of the inner product of the corresponding row and column of the factors that would appear in the usual matrix product. The implication is that support rank upper bounds on partial matrix multiplication tensors in this general sense give upper bounds on the support rank exponent of matrix multiplication.
\end{abstract}

\section{Introduction}

Understanding how the number of arithmetic operations required to multiply two $n\times n$ matrices scales with $n$ is a major open problem in theoretical computer science. The exponent of matrix multiplication $\omega$ is defined as the smallest real number such that, for all $\epsilon>0$, matrix multiplication can be performed using $O(n^{\omega+\epsilon})$ operations. It is known that at least $n^2$ operations are required, i.e., $2\le\omega$, while the basic algorithm that computes each entry of the product separately as inner products of rows and columns shows that $\omega\le 3$.

While the best lower bound is still $2\le\omega$, which is conjectured to be tight, the upper bound has been gradually improved through the development of increasingly more powerful tools. On a high level, upper bounds are proved by finding a starting tensor (bilinear map) $T$, an efficient way to compute $T$ (using few multiplications), and efficiently reducing large matrix multiplications to high tensor powers of $T$. The improved upper bounds have been made possible through a combination of enlarging the class of suitable tensors to which matrix multiplication can be reduced, finding specific good starting tensors and, more recently, refining the reduction step to improve its efficiency. We refer to \cite{blaser2013fast,pan2018fast} and \cite[Chapter 15]{burgisser1997algebraic} for more details on fast matrix multiplication algorithms.

The first upper bound below $3$ was Strassen's breakthrough result \cite{strassen1969gaussian}, in which he showed $\omega\le\log_2 7$ using simply the $2\times 2\times 2$ matrix multiplication tensor as a starting tensor. In \cite{pan1978strassen} Pan  introduced a technique called trilinear aggregation, for which the starting tensor is a possibly rectangular matrix multiplication, leading to the upper bound $\omega\le 2.79$. Bini et al. \cite{bini1979complexity} introduced the new concept of border rank (which, over the complex numbers, can be viewed as an arbitrary-precision approximation of a bilinear map), and showed that these can be used equally well for the purpose of bounding the asymptotic complexity, thereby relaxing what it means to efficiently compute the starting tensor $T$. Subsequently, Sch\"onehage showed how one can use partial matrix multiplication as a starting tensor \cite{schonhage1981partial}. This includes the case when the starting tensor is a sum of disjoint matrix multiplications, which comes with an even better upper bound known as the asymptotic sum inequality. In \cite{strassen1986asymptotic} Strassen introduced the laser method, which used a starting tensor that has a special block decomposition, with the nonzero blocks being matrix multiplication tensors, and the arrangement of the nonzero blocks (known as the outer structure) being also the same as in a matrix multiplication tensor (in the usual basis of matrix units). Subsequently, Coppersmith and Winograd developed a much more general version of the laser method, capable of reducing large matrix multiplications to powers of a starting tensor with an outer structure that is not a matrix multiplication, but what is known as a tight set, obtaining the upper bound $\omega\le 2.3755$ \cite{coppersmith1990matrix}. Since that paper, subsequent improvements did not result from enlarging the class of starting tensors and neither have better starting tensors been found. Instead, recent upper bounds are based on applying the laser method in more refined ways to tensor powers of the Coppersmith--Winograd tensor \cite{stothers2010complexity,williams2012multiplying,le2014algebraic,alman2021refined,duan2022faster,williams2024new}, gradually lowering to the current best value $\omega\le 2.371552$.

However, it has become clear that the current methods and some of their variants will not be able to prove a bound that is substantially closer to the conjectured value of $2$. The first such barrier result, proved by Ambainis, Filmus and Le Gall \cite{ambainis2015fast}, shows that it is not possible to prove $\omega\le 2.30\ldots$ using the Coppersmith--Winograd starting tensors with a range of methods. This has been extended to wider classes of methods, with lower bounds that are weaker but still strictly above $2$ \cite{christandl2021barriers,alman2018further,alman2018limits,alman2021limits} (with the exception of one of the ``small'' Coppersmith--Winograd tensors, which could still prove $\omega=2$ if its asymptotic rank was equal to $3$).

A different toolset is provided by the recent group-theoretic approach of Cohn and Umans \cite{cohn2003group,cohn2005group}, and its extension to adjacency algebras of coherent configurations \cite{cohn2013fast}. The group-theoretic approach proceeds by embedding a large matrix multiplication into multiplication in a group algebra over $\complexes$, which is itself a direct sum of several matrix multiplication tensors. Within this framework, it is possible to reproduce many of the previously known upper bounds. On the other hand, barriers exist for the group-theoretic approach as well \cite{blasiak2017cap,blasiak2017groups}.

The closely related approach of \cite{cohn2013fast} makes use of a new ingredient, another relaxation of the rank, called the support rank (see also in \cref{sec:supportrank}). The key insight is that matrix multiplication can also be asymptotically reduced to weighted matrix multiplication. To make this quantitative, one defines the s-rank exponent of matrix multiplication $\omega_s$ as the smallest number such that, for all $\epsilon>0$, $O(n^{\omega_s+\epsilon})$ arithmetic operations are sufficient for computing weighted matrix products of the form
\begin{equation}\label{eq:weightedmm}
\left(\sum_{j=1}^n\lambda_{i,j,k}A_{i,j}B_{j,k}\right)_{i,k\in[n]},
\end{equation}
where $(A_{i,j})_{i,j\in[n]}$ and $(B_{i,j})_{i,j\in[n]}$ are the input matrices and $\lambda_{i,j,k}$ are arbitrary nonzero numbers. Clearly, $\omega_s\le\omega$. In \cite{cohn2013fast} it is shown that $\omega\le(3\omega_s-2)/2$, which implies that weighted matrix multiplication tensors can also be used as starting tensors, with an overhead that vanishes as $\omega\to 2$ (equivalently: $\omega_s\to 2$). At present, the best upper bounds on $\omega_s$ come from upper bounds on $\omega$.

In this paper, we study a problem raised in \cite{schonhage1981partial} on more general partial matrix multiplication tensors. In that paper, Sch\"onhage considered the task of multiplying matrices that are partially filled, i.e., some of the entries in a given pattern are fixed to zeros. With $n\times n$ matrices, the number of nonzero entries in the basis of matrix units is a number $f<n^3$ determined by the two patterns. He showed that, if the product of such partially filled matrices can be computed using $l$ multiplications (in a bilinear computation), then $\omega\le 3\frac{\log l}{\log f}$ \cite[Theorem 4.1]{schonhage1981partial}. However, there is an asymmetry in this result, which is apparent in the equivalent trilinear formulation, which amounts to computing the trace of $ABC$, where $A$ and $B$ are partially filled matrices, while $C$ is a full $n\times n$ matrix. This asymmetry was pointed out in \cite[Section 8]{schonhage1981partial}, but the generalization of the result to this case was left as an open problem.

We show that it is possible to reduce large \emph{weighted} matrix multiplications to tensor powers of partial (and weighted) matrix multiplications, of the form \eqref{eq:weightedmm} with arbitrary coefficients. Our main result is that if the number of nonzero coefficients in \eqref{eq:weightedmm} is $f$, and the (support) rank of the corresponding tensor is $l$, then $\omega_s\le 3\frac{\log l}{\log f}$ (\cref{thm:srankexponentbound}). In particular, if one could show that $\omega_s=\omega$, then our result would yield the same bound as  \cite{schonhage1981partial} for the partial matrix multiplications considered in that paper, but at the same time apply for general partial matrix multiplications.

The proof is based on the observation that if $\Lambda\subseteq I\times J\times K$ is the pattern of nonzero coefficients and $f:I\to I'$, $g:J\to J'$, $h:K\to K'$ are arbitrary maps, then the support rank of partial matrix multiplication with pattern $\Lambda$ is an upper bound on the support rank of the partial matrix multiplication with pattern $(f\times g\times h)(\Lambda)$. Applying this to $\Lambda^n$ and randomly chosen $f,g,h$, we show that it is possible to simulate total (weighted) matrix multiplication with the number of nonzero entries growing as $\lvert\Lambda\rvert^{n-o(n)}$. Moreover, we characterize the trade-off between the rates $(a,b,c)$ such that weighted multiplication of $2^{an+o(n)}\times 2^{bn+o(n)}$ and $2^{bn+o(n)}\times 2^{cn+o(n)}$ is possible with this method, showing in particular that the bound on $\omega_s$ is the best possible within this framework.

We note that, while a direct sum of disjoint (partial or total) matrix multiplications is itself a partial matrix multiplication, in this special case a better bound on $\omega$ ($\omega_s$) is provided by the asymptotic sum inequality. In this sense, as noted in \cite[Section 8]{schonhage1981partial}, the obtained bounds do not behave in a continuous way, since disjointness of the pattern can be destroyed by adding just a single additional nonzero entry. To the best of our knowledge, it is still an open problem to find an interpolating result that behaves in a more robust way, and this discontinuity is even more pronounced for the more general partial matrix multiplication tensors considered in the present paper.

The outline of this paper is the following. In \cref{sec:preliminaries} we summarize the relevant concepts and notation related to tensors and from information theory. In \cref{sec:supportrank} we review the notion of support rank, and prove a basic monotonicity result relating the support ranks for different supports. In \cref{sec:partialmm} we specialize to partial matrix multiplication tensors, prove our main result, and briefly explain how it can be combined with the asymptotic sum inequality and with the laser method.

\section{Preliminaries}\label{sec:preliminaries}

\subsection{Tensors}

Throughout $\mathbb{F}$ will be some fixed but arbitrary infinite field. There exist several essentially equivalent ways to view tensors: as multilinear maps or forms, arrays of numbers, or elements of a tensor product of vector spaces. We adopt the convention that a tensor (of order $3$) is a trilinear form
\begin{equation}
T=\sum_{\substack{i\in I  \\  j\in J  \\  k\in K}}t_{i,j,k}x_iy_jz_k,
\end{equation}
where $t_{i,j,k}\in\mathbb{F}$ are the coefficients (or coordinates) of the tensor, and $\{x_i\}_{i\in I}$, $\{y_j\}_{j\in J}$, and $\{z_k\}_{k\in K}$ are formal variables indexed by finite index sets $I,J,K$. If two tensors differ only in a renaming of the variables (in such a way that the three sets of variables are not mixed), then the tensors are \emph{isomorphic}. As we are interested in tensor properties that are invariant under isomorphism, we will routinely replace a tensor with an isomorphic one when convenient.

The \emph{unit tensors} are the tensors
\begin{equation}
\unittensor{r}=\sum_{i=1}^r x_iy_iz_i.
\end{equation}
The \emph{direct sum} of tensors, $T_1\oplus T_2$, can be seen as the usual sum of trilinear forms provided that the sets of variables are disjoint (which can be assumed after replacing the tensors with isomorphic ones). The \emph{tensor product} of $T_1$ and $T_2$ can be obtained from the usual product as polynomials by replacing the products of the form $x_{i_1}x_{i_2}$ by new variables $x_{i_1,i_2}$, indexed by pairs, and similarly replacing $y_{j_1}y_{j_2}$ with $y_{j_1,j_2}$ and $z_{k_1}z_{k_2}$ with $z_{k_1,k_2}$.

Any tensor can be written (in a non-unique way) as a sum of \emph{simple tensors}, which arise as the product of three linear maps, one in each set of variables:
\begin{equation}
T=\sum_{m=1}^L \left(\sum_{i\in I}a^{(m)}_ix_i\right)\left(\sum_{j\in J}b^{(m)}_jy_j\right)\left(\sum_{k\in K}c^{(m)}_kz_k\right).
\end{equation}
The smallest possible $L$ for which such a decomposition exists is the \emph{rank} of $T$, which we denote by $\tensorrank(T)$. The rank is subadditive under the direct sum and submultiplicative under the tensor product. We define the \emph{asymptotic rank} of $T$ as $\asymptoticrank(T)=\lim_{n\to\infty}\sqrt[n]{\tensorrank(T^{\otimes n})}=\inf_{n\ge 1}\sqrt[n]{\tensorrank(T^{\otimes n})}$, where the existence of the limit and the equality with the infimum is due to the Fekete lemma.

The set of tensors of rank at most $r$ is in general not closed. Over an algebraically closed field, the \emph{border rank} $\borderrank(T)$ of a tensor $T$ is the smallest $r$ such that $T$ lies in the closure of the set of tensors of rank at most $r$. The border rank is also subadditive under the direct sum and submultiplicative under the tensor product, and leads to the same asymptotic quantity: $\lim_{n\to\infty}\sqrt[n]{\borderrank(T^{\otimes n})}=\inf_{n\ge 1}\sqrt[n]{\borderrank(T^{\otimes n})}=\asymptoticrank(T)$. The there quantities satisfy the inequalities $\asymptoticrank(T)\le\borderrank(T)\le\tensorrank(T)$.

The \emph{matrix multiplication tensors}, parametrized by positive integers $l,m,n$, are
\begin{equation}
\mmtensor{l}{m}{n}=\sum_{i=1}^l\sum_{j=1}^m\sum_{k=1}^n x_{i,j}y_{j,k}z_{k,i},
\end{equation}
and correspond to the multiplication of an $l\times m$ and an $m\times n$ matrix as a bilinear map. An important property is that $\mmtensor{l_1}{m_1}{n_1}\otimes\mmtensor{l_2}{m_2}{n_2}$ is isomorphic to $\mmtensor{l_1l_2}{m_1m_2}{n_1n_2}$, which can be understood as multiplication of partitioned matrices. The \emph{exponent of matrix multiplication} is $\omega=\log\asymptoticrank(\mmtensor{2}{2}{2})$, where the base of the logarithm is $2$. The value of $\omega$ is known to be between $2$ and $2.371552$ \cite{williams2024new}.

\subsection{Types}

We will make use of some of the fundamental combinatorial tools from information theory, more precisely the method of types. For a thorough introduction we refer to \cite[Chapter 2]{csiszar2011information}.

We denote the set of probability distributions on a finite set $\mathcal{X}$ by $\distributions(\mathcal{X})$. We use the notation $P(x)$ for the probability of $\{x\}$ under the measure $P$, and think of $P$ as a function $\mathcal{X}\to\nonnegativereals$. The \emph{support} of $P$ is $\support P=\setbuild{x\in\mathcal{X}}{P(x)\neq 0}$.

An \emph{$n$-type} (over $\mathcal{X}$) is an element of $\distributions(\mathcal{X})$ that assigns probabilites that are multiples of $\frac{1}{n}$ to the elements, i.e., such that $nP$ is integer valued. The set of $n$-types over $\mathcal{X}$ will be denoted by $\distributions[n](\mathcal{X})$. The number of $n$-types can be estimated as $\left\lvert\distributions[n](\mathcal{X})\right\rvert\le(n+1)^{\lvert\mathcal{X}\rvert}$ \cite[Lemma 2.2]{csiszar2011information}.

Given $P\in\distributions[n](\mathcal{X})$, we can consider the set of strings in $\mathcal{X}^n$ in which each element $x\in\mathcal{X}$ occurs exactly $nP(x)$ times. This set is the \emph{type class} $\typeclass{n}{P}$. Another way to view type classes is that they are precisely the orbits in $\mathcal{X}^n$ under the action of the symmetric group $S_n$ permuting the factors. The cardinality of a type class can be estimated as \cite[Lemma 2.3]{csiszar2011information}
\begin{equation}
\frac{1}{(n+1)^{\lvert\mathcal{X}\rvert}}2^{n\entropy(P)}\le\lvert\typeclass{n}{P}\rvert=\binom{n}{nP}=\frac{n!}{\prod_{x\in\mathcal{X}}(nP(x))!}\le 2^{n\entropy(P)},
\end{equation}
where
\begin{equation}
\entropy(P)=-\sum_{x\in\mathcal{X}}P(x)\log P(x)
\end{equation}
is the (Shannon) \emph{entropy} measured in bits (i.e., with logarithm to base $2$). Here we adopt the convention that $P(x)\log P(x)=0$ when $P(x)=0$.

If $P$ is a distribution on (a subset of) a product set, e.g., $P\in\distributions(I\times J\times K)$, then we can form its \emph{marginals}, for instance
\begin{align}
P_I(i) & = \sum_{\substack{j\in J  \\  k\in K}}P(i,j,k)  \\
\intertext{and}
P_{IJ}(i,j) & = \sum_{k\in K}P(i,j,k).
\end{align}
In this setting it is often convenient to use notations such as $\entropy(I)_P=\entropy(P_I)$ and $\entropy(IJ)_P=\entropy(P_{IJ})$ for the marginals. In particular, $\entropy(P)=\entropy(IJK)_P$ for $P\in\distributions(I\times J\times K)$. We also define the \emph{conditional entropy} $\entropy(I|J)_P=\entropy(IJ)_P-\entropy(J)_P$, $\entropy(K|IJ)_P=\entropy(IJK)_P-\entropy(IJ)_P$, etc. Entropies and conditional entropies are nonnegative.

The conditional entropy appears when estimating the cardinality of the fibers of the canonical surjection map $\pi:\typeclass{n}{P_{IJ}}\to\typeclass{n}{P_J}$ (obtained by restricting the map $(I\times J)^n=I^n\times J^n\to J_n$). Since the map is $S_n$ equivariant and $\typeclass{n}{P_J}$ is a single orbit, every fiber $\pi^{-1}(i)$ has the same cardinality, equal to $\lvert\typeclass{n}{P_{IJ}}\rvert/\lvert\typeclass{n}{P_J}\rvert\approx 2^{n\entropy(I|J)_P}$, where the approximation means ignoring factors polynomial in $n$.

\section{Support rank}\label{sec:supportrank}

We follow \cite[Section 3]{cohn2013fast} except that it seems more appropriate to define the support rank as a property of the support (as opposed to a tensor).
\begin{definition}
Let $\Phi\subseteq I\times J\times K$, where $I$, $J$, and $K$ are finite sets. If
\begin{equation}
T=\sum_{(i,j,k)\in\Phi} t_{i,j,k}x_iy_jz_k,
\end{equation}
where the coeffcients satisfy $t_{i,j,k}\neq 0$ for all $(i,j,k)\in\Phi$, then we say that the \emph{support} of $T$ is $\Phi$, and write $\support(T)=\Phi$. The \emph{support rank} of $\Phi$ is
\begin{equation}
\supportrank(\Phi)=\min\setbuild{\tensorrank(T)}{\support(T)=\Phi}.
\end{equation}

In a similar vein, the \emph{border support rank} of $\Phi$ is
\begin{equation}
\bordersupportrank(\Phi)=\min\setbuild{\borderrank(T)}{\support(T)=\Phi}.
\end{equation}
\end{definition}

If $\Phi_1\subseteq I_1\times J_1\times K_1$ and $\Phi_2\subseteq I_2\times J_2\times K_2$, then we can form the sum $\Phi_1+\Phi_2\subseteq(I_1\sqcup I_2)\times(J_1\sqcup J_2)\times(K_1\sqcup K_2)$ by embedding both subsets in the block diagonal, and the product $\Phi_1\times\Phi_2\subseteq(I_1\times I_2)\times(J_1\times J_2)\times(K_1\times K_2)$ by rearranging and grouping the factors as indicated. Clearly, $\support(T_1\oplus T_2)=\support(T_1)+\support(T_2)$ and $\support(T_1\otimes T_2)=\support(T_1)\times\support(T_2)$, and it follows that the support rank and the border support rank are both subadditive and submultiplicative. In analogy with the tensor rank, we can therefore introduce the \emph{asymptotic support rank}
\begin{equation}
\asymptoticsupportrank(\Phi)=\lim_{n\to\infty}\sqrt[n]{\supportrank(\Phi)}=\lim_{n\to\infty}\sqrt[n]{\bordersupportrank(\Phi)}=\inf_{n\ge 1}\sqrt[n]{\supportrank(\Phi)}=\inf_{n\ge 1}\sqrt[n]{\bordersupportrank(\Phi)}.
\end{equation}

In particular, the \emph{support rank exponent of matrix multiplication} (or exponent of weighted matrix multiplication) is $\omega_2=\log\asymptoticsupportrank(\support(\mmtensor{2}{2}{2}))$. It is clear that $2\le\omega_s\le\omega$ (more generally, for any tensor $T$ we have $\asymptoticsupportrank(\support(T))\le\asymptoticrank(T)$). Cohn and Umans proved that the two exponents satisfy the inequality $\omega-2\le\frac{3}{2}(\omega_s-2)$ \cite[Theorem 3.6]{cohn2013fast}. In particular, $\omega_s=2$ if and only if $\omega=2$.

More generally, some applications require multiplying $n^a\times n^b$ and $n^b\times n^c$ matrices, where $a,b,c\ge 0$ are given constants, as $n\to\infty$. In this setting, one defines the exponents
\begin{equation}
\omega(a,b,c)=\lim_{n\to\infty}\frac{1}{n}\log\tensorrank(\mmtensor{\lfloor 2^{an}\rfloor}{\lfloor 2^{bn}\rfloor}{\lfloor 2^{cn}\rfloor})
\end{equation}
and
\begin{equation}
\omega_s(a,b,c)=\lim_{n\to\infty}\frac{1}{n}\log\supportrank(\support(\mmtensor{\lfloor 2^{an}\rfloor}{\lfloor 2^{bn}\rfloor}{\lfloor 2^{cn}\rfloor})).
\end{equation}
Both are homogeneous of degree $1$, jointly convex, and symmetric in the three arguments, therefore by symmetrization we obtain the inequalities $\omega=\omega(1,1,1)\le\frac{3}{a+b+c}\omega(a,b,c)$ and $\omega_s=\omega_s(1,1,1)\le\frac{3}{a+b+c}\omega_s(a,b,c)$. Clearly, $\omega_s(a,b,c)\le\omega(a,b,c)$ for all $a,b,c\ge 0$. We are not aware of any bound in the other direction except the aforementioned one between $\omega(1,1,1)=\omega$ and $\omega_s(1,1,1)=\omega_s$.

The tensor rank is nonincreasung under tensor restriction (specialization), and the border rank is nonincreasing under tensor degeneration. An analogous notion on the level of supports is combinatorial degeneration \cite[(15.29) Definition]{burgisser1997algebraic}, and it follows from a standard proposition that border support rank is nonincreasing under combinatorial degeneration \cite[(15.30) Proposition]{burgisser1997algebraic}. In the following proposition we show that the support rank is nonincreasing when taking direct images under product maps.
\begin{proposition}\label{prop:supportrankproductmap}
Let $I$, $J$, $K$, $I'$, $J'$, and $K'$ be finite sets, $\Phi\subseteq I\times J\times K$, and let $f:I\to I'$, $g:J\to J'$, and $h:K\to K'$ be functions. Then $\supportrank(\Phi)\ge\supportrank((f\times g\times h)(\Phi))$.
\end{proposition}
\begin{proof}
Let
\begin{equation}
T=\sum_{(i,j,k)\in\Phi} t_{i,j,k}x_iy_jz_k
\end{equation}
be a tensor such that the coefficients $t_{i,j,k}$ are nonzero for all $(i,j,k)\in\Phi$ and $\tensorrank(T)=\supportrank(\Phi)$. We consider the following restrictions of $T$, where $a_i,b_j,c_k\in\mathbb{F}$:
\begin{equation}
\begin{split}
T'
 & = \sum_{(i,j,k)\in\Phi} t_{i,j,k}(a_i x_{f(i)})(b_j y_{g(i)})(c_k z_{f(k)})  \\
 & = \sum_{i'\in I'}\sum_{j'\in J'}\sum_{k'\in K'}\sum_{(i,j,k)\in\Phi\cap(f^{-1}(i')\times g^{-1}(j')\times h^{-1}(k'))} (a_i b_j c_k t_{i,j,k})x_{i'}y_{j'}z_{k'}
\end{split}
\end{equation}
For any choice of the field elements $a_i,b_j,c_k$ we have $\tensorrank(T')\le\tensorrank(T)$ (since $T'$ is a restriction of $T$), and $\support(T')\subseteq (f\times g\times h)(\Phi)$ by construction. Therefore the claim follows if there is a choice such that $\support(T')=(f\times g\times h)(\Phi)$.

Note that for all $(i',j',k')\in(f\times g\times h)(\Phi)$, the coefficient of $x_{i'}y_{j'}z_{k'}$ is a (cubic) polynomial in the indeterminates $a_i,b_j,c_k$ ($i\in I$, $j\in J$, $k\in K$). These polynomials do not vanish identically: for $(i',j',k')\in(f\times g\times h)(\Phi)$ and $(i_0,j_0,k_0)\in\Phi\cap(f^{-1}(i')\times g^{-1}(j')\times h^{-1}(k'))$, we can set $a_i=1$ if $i=i_0$ and $a_i=0$ otherwise, and similarly for $b_j$ and $c_k$, then the coefficient of $(i',j',k')$ is equal to $t_{i_0,j_0,k_0}\neq 0$. By the Schwartz--Zippel lemma (recall our assumption that $\lvert\mathbb{F}\rvert=\infty$), we can conclude that there is a choice of $a_i,b_j,c_k$ such that none of these polynomials vanish, i.e., $\support(T')=(f\times g\times h)(\Phi)$.
\end{proof}

\section{Partial matrix multiplication}\label{sec:partialmm}

Sch\"onhage introduced partial matrix multiplication in the sense of multiplication of matrices with some of the entries filled with zeros \cite{schonhage1981partial}. Such an operation is specified by positive integers $l,m,n$ and a pair of patterns $A\subseteq[l]\times[m]$ and $B\subseteq[j]\times[k]$. Using the indicator function
\begin{equation}
1_A(i,j)=\begin{cases}
1 & \text{if $(i,j)\in A$}  \\
0 & \text{otherwise,}
\end{cases}
\end{equation}
and the similarly defined $1_B$, the corresponding tensor is written as
\begin{equation}
\partialmmtensor{A,B}=\sum_{i=1}^l\sum_{j=1}^m\sum_{k=1}^n 1_A(i,j)1_B(j,k)x_{i,j}y_{j,k}z_{k,i}.
\end{equation}
The total number of terms is
\begin{equation}
f=\sum_{i=1}^l\sum_{j=1}^m\sum_{k=1}^n 1_A(i,j)1_B(j,k),
\end{equation}
which is also an upper bound on the tensor rank, corresponding to the straightforward algorithm computing the partial matrix product (when $A=[l]\times[m]$ and $B=[j]\times[k]$, this number is $lmn$). Partial matrix multiplications are useful for bounding the exponent of matrix multiplication by the following result:
\begin{theorem}[Sch\"onhage, {\cite[Theorem 4.1.]{schonhage1981partial}}]
If $\partialmmtensor{A,B}$ is a partial matrix multiplication tensor containing $f$ terms, then $\omega\le 3\frac{\log\tensorrank(\partialmmtensor{A,B})}{\log f}$.
\end{theorem}

We consider partial matrix multiplications in the following more general sense. Let $l,m,n$ be positive integers and $\Lambda\subseteq[l]\times[m]\times[n]$. We define
\begin{equation}
\partialmmtensor{\Lambda}=\sum_{i=1}^l\sum_{j=1}^m\sum_{k=1}^n 1_\Lambda(i,j,k)x_{i,j}y_{j,k}z_{k,i}.
\end{equation}
It is not known if the above theorem can be generalized to an upper bound on $\omega$ in terms of the tensors $\partialmmtensor{\Lambda}$.
A special case, highlighted in \cite[Section 8]{schonhage1981partial} as an open problem, is to find a specified subset of the entries of a matrix product, where the factors are partially filled with zeros (equivalently: to compute the trace of a product of three partially filled matrices).

We introduce the following notation for the support of a partial matrix multiplication tensor:
\begin{equation}\label{eq:partialmmsupportdef}
\begin{split}
\partialmmsupport{\Lambda}
 & = \support(\partialmmtensor{\Lambda})  \\
 & = \setbuild{((i,j),(j,k),(k,i))\in(I\times J)\times(J\times K)\times(K\times I)}{(i,j,k)\in\Lambda}.
\end{split}
\end{equation}
If $\Lambda_1\subseteq I_1\times J_1\times K_1$ and $\Lambda_2\subseteq I_2\times J_2\times K_2$, then we may form $\Lambda_1\times\Lambda_2\in(I_1\times I_2)\times(J_1\times J_2)\times(K_1\times K_2)$ (rearranging the factors similarly as with the supports), and we clearly have $\partialmmtensor{\Lambda_1\times\Lambda_2}=\partialmmtensor{\Lambda_1}\otimes\partialmmtensor{\Lambda_2}$ and $\partialmmsupport{\Lambda_1\times\Lambda_2}=\partialmmsupport{\Lambda_1}\times\partialmmsupport{\Lambda_2}$ (the sum operations behave in a similar way). Total matrix multiplication is the special case when $\Lambda=I\times J\times K$. Combining this fact with the notations from the previous section, the support rank exponent of matrix multiplication can be characterized as $\omega_s=\log\asymptoticsupportrank(\partialmmsupport{[2]\times[2]\times[2]})$.

The construction $\Lambda\mapsto\partialmmsupport{\Lambda}$ behaves well with respect to product maps as well, which can be seen directly from \eqref{eq:partialmmsupportdef}:
\begin{proposition}\label{prop:partialmmproductmap}
Let $I$, $J$, $K$, $I'$, $J'$, and $K'$ be finite sets, $\Lambda\subseteq I\times J\times K$, and let $f:I\to I'$, $g:J\to J'$, and $h:K\to K'$ be functions. Then
\begin{equation}
\partialmmsupport{(f\times g\times h)(\Lambda)}
 = ((f\times g)\times(g\times h)\times(h\times f))(\partialmmsupport{\Lambda}).
\end{equation}
\end{proposition}

\begin{corollary}\label{cor:partialmmproductmapsrank}
In the setting of \cref{prop:partialmmproductmap}, $\supportrank(\partialmmsupport{\Lambda})\ge\supportrank(\partialmmsupport{(f\times g\times h)(\Lambda)})$.
\end{corollary}
\begin{proof}
The claim follows immediately from \cref{prop:supportrankproductmap,prop:partialmmproductmap}.
\end{proof}
This inequality means that in order to asymptotically reduce weighted matrix multiplications to partial (and possibly also weighted) matrix multiplications corresponding to powers of a given pattern $\Lambda$, it is sufficient to find maps $f_n:I^n\to A_n$, $g_n:J^n\to B_n$, and $h_n:K^n\to C_n$ such that $(f_n\times g_n\times h_n)(\Lambda^n)=A_n\times B_n\times C_n$ and the sets $A_n$, $B_n$, and $C_n$ are exponentially large. In the following we determine the precise trade-off between the exponents that are compatible with the existence of such maps.
\begin{definition}
A rate triple $(a,b,c)$ is \emph{achievable} for the pattern $\Lambda\subseteq I\times J\times K$, if there exist sequences of maps $f_n:I^n\to [\lfloor 2^{an}\rfloor]$, $g_n:J^n\to [\lfloor 2^{bn}\rfloor]$, and $h_n:K^n\to [\lfloor 2^{cn}\rfloor]$ such that $(f_n\times g_n\times h_n)(\Lambda^n)=[\lfloor 2^{an}\rfloor]\times[\lfloor 2^{bn}\rfloor]\times[\lfloor 2^{cn}\rfloor]$ for all $n$.

The \emph{capacity region} $C(\Lambda)\subseteq\nonnegativereals^3$ for $\Lambda$ is the closure of the set of achievable triples.
\end{definition}

One can see that the capacity region is convex by a time sharing argument. Indeed, suppose that $f_n:I^n\to [\lfloor 2^{an}\rfloor]$, $g_n:J^n\to [\lfloor 2^{bn}\rfloor]$, and $h_n:K^n\to [\lfloor 2^{cn}\rfloor]$ and $f'_n:I^n\to [\lfloor 2^{a'n}\rfloor]$, $g'_n:J^n\to [\lfloor 2^{b'n}\rfloor]$, and $h'_n:K^n\to [\lfloor 2^{c'n}\rfloor]$ are sequences of maps such that $(f_n\times g_n\times h_n)(\Lambda^n)=[\lfloor 2^{an}\rfloor]\times[\lfloor 2^{bn}\rfloor]\times[\lfloor 2^{cn}\rfloor]$ and $(f'_n\times g'_n\times h'_n)(\Lambda^n)=[\lfloor 2^{a'n}\rfloor]\times[\lfloor 2^{b'n}\rfloor]\times[\lfloor 2^{c'n}\rfloor]$ for all $n$. For every $\lambda\in[0,1]$, we can form the maps $f_{\lfloor\lambda n\rfloor}\times f'_{n-\lfloor\lambda n\rfloor}:I^n\to [\lfloor 2^{a\lfloor\lambda n\rfloor}\rfloor]\times[\lfloor 2^{a'(n-\lfloor n\rfloor)}\rfloor]$, etc., so that the product of the three similarly constructed maps map $\Lambda^n$ onto a product of size approximately $2^{n(\lambda a+(1-\lambda)a')}\times 2^{n(\lambda b+(1-\lambda)b')}\times 2^{n(\lambda a+(1-\lambda) c')}$.

\begin{lemma}\label{lem:singleshotproductmap}
Let $X,Y,Z$ be finite sets, $T\subseteq X\times Y\times Z$, and let $A\subset X$, $B\subseteq Y$, $C\subseteq Z$ be independent random subsets such that elements of $X$ are included in $A$ with the same probability $\alpha$, independently of each other and of the other subsets, and similarly, every element of $Y$ is in $B$ with probability $\beta$, and every element of $Z$ is in $C$ with probability $\gamma$. Then
\begin{equation}
\begin{split}
\probability\left[T\cap(A\times B\times C)=\emptyset\right]
 & \le (1-\alpha)^{\lvert\pi_X(T)\rvert}  \\
 & \qquad +(1-\beta)^{\min_{x\in\pi_X(T)}\lvert\pi_Y(T\cap(\{x\}\times Y\times Z))\rvert}  \\
 & \qquad +(1-\gamma)^{\min_{(x,y)\in\pi_{XY}(T)}\lvert\pi_Z(T\cap(\{x\}\times\{y\}\times Z))\rvert},
\end{split}
\end{equation}
where, e.g., $\pi_X:X\times Y\times Z\to X$ and $\pi_{XY}:X\times Y\to Z\to X\times Y$ denote the projection maps.
\end{lemma}
\begin{proof}
Let $p=\probability\left[T\cap(A\times B\times C)\neq\emptyset\right]$.
\begin{equation}
\begin{split}
p
 & = \sum_{\substack{S\subseteq X  \\  \pi_X(T)\cap S\neq\emptyset}}\probability\left[A=S\right]\probability\left[T\cap(S\times B\times C)\neq\emptyset\right]  \\
 & \ge \sum_{\substack{S\subseteq X  \\  \pi_X(T)\cap S\neq\emptyset}}\probability\left[A=S\right]
       \min_{x\in\pi_X(T)}\probability\left[T\cap(\{x\}\times B\times C)\neq\emptyset\right]  \\
 & = \probability\left[\pi_X(T)\cap A\neq\emptyset\right]
     \min_{x\in\pi_X(T)}\sum_{\substack{S\subseteq Y  \\  \pi_Y(T\cap(\{x\}\times Y\times Z))  \\  \cap S\neq\emptyset}}\probability\left[B=S\right]\probability\left[T\cap(\{x\}\times S\times C)\neq\emptyset\right]  \\
 & \ge \probability\left[\pi_X(T)\cap A\neq\emptyset\right]
       \min_{x\in\pi_X(T)}\sum_{\substack{S\subseteq Y  \\  \pi_Y(T\cap(\{x\}\times Y\times Z))  \\  \cap S\neq\emptyset}}\probability\left[B=S\right]  \\
   &\qquad\min_{y\in\pi_Y(T\cap(\{x\}\times Y\times Z))}\probability\left[T\cap(\{x\}\times \{y\}\times C)\neq\emptyset\right]  \\
 & \ge \probability\left[\pi_X(T)\cap A\neq\emptyset\right]
       \min_{x\in\pi_X(T)}\probability\left[\pi_Y(T\cap(\{x\}\times Y\times Z))\cap B\neq\emptyset\right]  \\
   &\qquad\min_{(x,y)\in\pi_{XY}(T)}\probability\left[\pi_Z(T\cap(\{x\}\times \{y\}\times Z))\cap C\neq\emptyset\right]  \\
 & = \left(1-(1-\alpha)^{\lvert\pi_X(T)\rvert}\right)
     \left(1-(1-\beta)^{\min_{x\in\pi_X(T)}\lvert\pi_Y(T\cap(\{x\}\times Y\times Z))\rvert}\right)  \\
   & \qquad\left(1-(1-\gamma)^{\min_{(x,y)\in\pi_{XY}(T)}\lvert\pi_Z(T\cap(\{x\}\times\{y\}\times Z))\rvert}\right)  \\
 & \ge 1-(1-\alpha)^{\lvert\pi_X(T)\rvert}-(1-\beta)^{\min_{x\in\pi_X(T)}\lvert\pi_Y(T\cap(\{x\}\times Y\times Z))\rvert}  \\
       &\qquad-(1-\gamma)^{\min_{(x,y)\in\pi_{XY}(T)}\lvert\pi_Z(T\cap(\{x\}\times\{y\}\times Z))\rvert}
\end{split}
\end{equation}
\end{proof}

\begin{lemma}\label{lem:subsetprojection}
Let $X$ and $Y$ be finite sets, $S \subseteq X\times Y$, and let $\pi_X:X\times Y\to X$ denote the projection onto the first factor. Then
\begin{equation}
\frac{\lvert S\rvert}{\lvert X\times Y\rvert}\le\frac{\lvert\pi_X(S)\rvert}{\lvert X\rvert}.
\end{equation}
\end{lemma}
\begin{proof}
Since $S\subseteq\pi_X^{-1}(\pi_X(S))$, we have $\lvert S\rvert\le\lvert\pi_X^{-1}(\pi_X(S))\rvert=\lvert Y\rvert\lvert\pi_X(S)\rvert$.
\end{proof}

\begin{proposition}\label{prop:capacityregion}
Let $\Lambda\subseteq I\times J\times K$. The capacity region for $\Lambda$ is the set of all triples $(a,b,c)$ satisfying
\begin{subequations}
\begin{align}
a & \le \entropy(I)_P  \label{eq:entropyconditionI}  \\
b & \le \entropy(J)_P  \label{eq:entropyconditionJ}  \\
c & \le \entropy(K)_P  \label{eq:entropyconditionK}  \\
a+b & \le \entropy(IJ)_P  \label{eq:entropyconditionIJ}  \\
a+c & \le \entropy(IK)_P  \label{eq:entropyconditionIK}  \\
b+c & \le \entropy(JK)_P  \label{eq:entropyconditionJK}  \\
a+b+c & \le \entropy(IJK)_P,  \label{eq:entropyconditionIJK}
\end{align}
\end{subequations}
for some probability distribution $P\in\distributions(\Lambda)$ (that may depend on $(a,b,c)$).
\end{proposition}
\begin{proof}
Let $P\in\distributions(\Lambda)$, $\delta>0$, $a=\max\{\entropy(I)_P-\delta,0\}$, $b=\max\{\entropy(J|I)_P-\delta,0\}$, and $c=\max\{\entropy(K|IJ)_P-\delta,0\}$. Choose $n$-types $P_n\in\distributions[n](\Lambda)$ such that $\support(P_n)=\support(P)$ and $\lim_{n\to\infty}P_n=P$.

Draw the maps $f_n:I^n\to[\lfloor 2^{an}\rfloor]$, $g_n:J^n\to [\lfloor 2^{bn}\rfloor]$, and $h_n:K^n\to [\lfloor 2^{cn}\rfloor]$ at random, independently and uniformly distributed on the set of all maps. For $(i',j',k')\in[\lfloor 2^{an}\rfloor]\times[\lfloor 2^{bn}\rfloor]\times[\lfloor 2^{cn}\rfloor]$, we have
\begin{equation}
\begin{split}
\probability\left[(i',j',k')\notin(f_n\times g_n\times h_n)(\Lambda^n)\right]
 & \le \probability\left[(i',j',k')\notin(f_n\times g_n\times h_n)(\typeclass{n}{P_n})\right]  \\
 & = \probability\left[\typeclass{n}{P_n}\cap(f_n^{-1}(i')\times g_n^{-1}(j')\times h_n^{-1}(k'))=\emptyset\right]
\end{split}
\end{equation}
By the choice of the random maps $f_n$, $g_n$, and $h_n$, the subsets $f_n^{-1}(i')\subseteq I^n$, $g_n^{-1}(j')\subseteq J^n$, and $h_n^{-1}(k')\subseteq K^n$ are distributed as in \cref{lem:singleshotproductmap}, with $\alpha=\lfloor 2^{an}\rfloor^{-1}$, $\beta=\lfloor 2^{bn}\rfloor^{-1}$, and $\gamma=\lfloor 2^{cn}\rfloor^{-1}$. Since $S_n$ acts transitively on $\typeclass{n}{P_n}$ and the projection maps to $I^n$ and $I^n\times J^n$ are equivariant, the cardinality of each fiber is the same, and equal to the cardinality of $\typeclass{n}{P_n}$ divided by the cardinality of the image, which is itself a type class with type equal to the corresponding marginal. From \cref{lem:singleshotproductmap} we conclude that the probability $p$ that $(i',j',k')$ is not in the image of $\Lambda^n$ is at most
\begin{equation}
\begin{split}
p
 & \le (1-2^{-an})^{\left\lvert\typeclass{n}{(P_n)_I}\right\rvert}
 +(1-2^{-bn})^{\frac{\left\lvert\typeclass{n}{(P_n)_{IJ}}\right\rvert}{\left\lvert\typeclass{n}{(P_n)_I}\right\rvert}}
 +(1-2^{-cn})^{\frac{\left\lvert\typeclass{n}{(P_n)_{IJK}}\right\rvert}{\left\lvert\typeclass{n}{(P_n)_{IJ}}\right\rvert}}  \\
 & \le e^{-2^{-an}2^{n\entropy(I)_{P_n}}(n+1)^{-\lvert I\rvert}}
     + e^{-2^{-bn}2^{n\entropy(J|I)_{P_n}}(n+1)^{-\lvert I\rvert\lvert J\rvert}}  \\
    &\qquad+ e^{-2^{-cn}2^{n\entropy(K|IJ)_{P_n}}(n+1)^{-\lvert I\rvert\lvert J\rvert\lvert K\rvert}},
\end{split}
\end{equation}
which goes to zero like $e^{-2^{\Omega(n)}}$. The probability that there is at least one triple $(i',j',k')\in[\lfloor 2^{an}\rfloor]\times[\lfloor 2^{bn}\rfloor]\times[\lfloor 2^{cn}\rfloor]$ not in the image $(f_n\times g_n\times h_n)(\Lambda^n)$ is therefore at most
\begin{equation}
2^{(a+b+c)n}e^{-2^{\Omega(n)}}.
\end{equation}
It follows that for all sufficiently large $n$, there exists a choice of maps $f_n$, $g_n$, and $h_n$ such that $(f_n\times g_n\times h_n)(\Lambda^n)=[\lfloor 2^{an}\rfloor]\times[\lfloor 2^{bn}\rfloor]\times[\lfloor 2^{cn}\rfloor]$, i.e., $(a,b,c)$ is achievable. Letting $\delta\to 0$, we conclude that the rate triple $(\entropy(I)_P,\entropy(J|I)_P,\entropy(K|IJ)_P)$ is achievable.

The reasoning above distinguishes an ordering of the three factors, which could have been arbitrary. By considering all possible permutations, can see that the following rate triples are achievable:
\begin{subequations}
\begin{gather}
(\entropy(I)_P,\entropy(J|I)_P,\entropy(K|IJ)_P)  \label{eq:vertexIJK}  \\
(\entropy(I)_P,\entropy(J|IK)_P,\entropy(K|I)_P)  \label{eq:vertexIKJ}  \\
(\entropy(I|J)_P,\entropy(J)_P,\entropy(K|IJ)_P)  \label{eq:vertexJIK}  \\
(\entropy(I|JK)_P,\entropy(J)_P,\entropy(K|J)_P)  \label{eq:vertexJKI}  \\
(\entropy(I|K)_P,\entropy(J|IK)_P,\entropy(K)_P)  \label{eq:vertexKIJ}  \\
(\entropy(I|JK)_P,\entropy(J|K)_P,\entropy(K)_P). \label{eq:vertexKJI}
\end{gather}
\end{subequations}
By the observation above, the convex hull of such points is achievable, as well as any point that can be obtained by decreasing the coordinates (as long as they stay nonnegative). We claim that the convex body obtained in this way is the same as the one determined by \labelcref{eq:entropyconditionI,eq:entropyconditionJ,eq:entropyconditionK,eq:entropyconditionIJ,eq:entropyconditionIK,eq:entropyconditionJK,eq:entropyconditionIJK}.

To show this, let us consider the problem of maximizing the linear functional $t_aa+t_bb+t_cc$ over the solution set of \labelcref{eq:entropyconditionI,eq:entropyconditionJ,eq:entropyconditionK,eq:entropyconditionIJ,eq:entropyconditionIK,eq:entropyconditionJK,eq:entropyconditionIJK} (with $a,b,c\ge 0$ understood), where $t_a,t_b,t_c\in\nonnegativereals$ are parameters. The dual linear program is to minimize
\begin{equation}
\entropy(I)_Pu_I+\entropy(J)_Pu_J+\entropy(K)_Pu_K+\entropy(IJ)_Pu_{IJ}+\entropy(IK)_Pu_{IK}+\entropy(JK)_Pu_{JK}+\entropy(IJK)_Pu_{IJK}
\end{equation}
over the nonnegative variables $u_I,u_J,\dots,u_{IJK}$, subject to
\begin{equation}
\begin{bmatrix}
1 & 0 & 0  &  0 & 1 & 1  &  1  \\
0 & 1 & 0  &  1 & 0 & 1  &  1  \\
0 & 0 & 1  &  1 & 1 & 0  &  1
\end{bmatrix}
\begin{bmatrix}
u_I  \\
u_J  \\
u_K  \\
u_{JK}  \\
u_{IK}  \\
u_{IJ}  \\
u_{IJK}
\end{bmatrix}\ge
\begin{bmatrix}
t_a  \\
t_b  \\
t_c
\end{bmatrix}.
\end{equation}
We can verify that the maximum in the primal problem is always attained at one of the points \labelcref{eq:vertexIJK,eq:vertexIKJ,eq:vertexJIK,eq:vertexJKI,eq:vertexKIJ,eq:vertexKJI}, depending on the ordering of $t_a,t_b,t_c$, by finding a feasible solution to the dual program where the objective function takes the same value. For instance, if $t_a\ge t_b\ge t_c$, then at the point $(\entropy(I)_P,\entropy(J|I)_P,\entropy(K|IJ)_P)$ the objective function evaluates to 
\begin{equation}
t_a\entropy(I)_P+t_b\entropy(J|I)_P+t_c\entropy(K|IJ)_P
 = (t_a-t_b)\entropy(I)_P+(t_b-t_c)\entropy(IJ)_P+t_c\entropy(IJK)_P,
\end{equation}
which is attained in the dual program at $u_I=t_a-t_c$, $u_{IJ}=t_b-t_c$, $u_{IJK}=t_c$, $u_J=u_K=u_{IK}=u_{JK}=0$, while if $t_a\ge t_c\ge t_b$, then maximum is attained at the vertex \eqref{eq:vertexIKJ} with dual optimal solution $u_I=t_a-t_c$, $u_{IK}=t_c-t_b$, $u_{IJK}=b$, $u_J=u_K=u_{IJ}=u_{JK}=0$, etc.

For the necessity of the conditions, suppose that $f_n:I^n\to [\lfloor 2^{an}\rfloor]$, $g_n:J^n\to [\lfloor 2^{bn}\rfloor]$, and $h_n:K^n\to [\lfloor 2^{cn}\rfloor]$ are sequences of maps such that $(f_n\times g_n\times h_n)(\Lambda^n)=[\lfloor 2^{an}\rfloor]\times[\lfloor 2^{bn}\rfloor]\times[\lfloor 2^{cn}\rfloor]$ for all $n$. Then
\begin{equation}
\begin{split}
[\lfloor 2^{an}\rfloor]\times[\lfloor 2^{bn}\rfloor]\times[\lfloor 2^{cn}\rfloor]
 & = (f_n\times g_n\times h_n)(\Lambda^n)  \\
 & = (f_n\times g_n\times h_n)\left(\bigcup_{P\in\distributions[n](\Lambda)}\typeclass{n}{P}\right)  \\
 & = \bigcup_{P\in\distributions[n](\Lambda)}(f_n\times g_n\times h_n)\left(\typeclass{n}{P}\right).
\end{split}
\end{equation}
For each $n$, let $P_n$ be a type class such that $\lvert(f_n\times g_n\times h_n)(\typeclass{n}{P_n})\rvert$ is maximal. Since the number of type classes is at most $(n+1)^{\lvert\Lambda\rvert}$, $\lvert(f_n\times g_n\times h_n)(\typeclass{n}{P_n})\rvert\ge(n+1)^{-\lvert\Lambda\rvert}\lfloor 2^{an}\rfloor\lfloor 2^{bn}\rfloor\lfloor 2^{cn}\rfloor$. On the other hand, $\lvert(f_n\times g_n\times h_n)(\typeclass{n}{P_n})\rvert\le\lvert\typeclass{n}{P_n}\rvert$, therefore
\begin{equation}
\frac{1}{n}\log\left[(n+1)^{-\lvert\Lambda\rvert}\lfloor 2^{an}\rfloor\lfloor 2^{bn}\rfloor\lfloor 2^{cn}\rfloor\right]\le\frac{1}{n}\log\lvert\typeclass{n}{P_n}\rvert.
\end{equation}
By considering the projections and using \cref{lem:subsetprojection}, we similarly obtain e.g., 
\begin{equation}
\frac{1}{n}\log\left[(n+1)^{-\lvert\Lambda\rvert}\lfloor 2^{an}\rfloor\lfloor 2^{bn}\rfloor\right]\le\frac{1}{n}\log\lvert\typeclass{n}{(P_n)_{IJ}}\rvert
\end{equation}
and
\begin{equation}
\frac{1}{n}\log\left[(n+1)^{-\lvert\Lambda\rvert}\lfloor 2^{an}\rfloor\right]\le\frac{1}{n}\log\lvert\typeclass{n}{(P_n)_I}\rvert.
\end{equation}
By compactness, we can choose a subsequence of distributions $P_{n_k}$ converging to some $P\in\distributions(\Lambda)$. Taking the limits of the inequalities along this subsequence gives the inequalities \labelcref{eq:entropyconditionI,eq:entropyconditionJ,eq:entropyconditionK,eq:entropyconditionIJ,eq:entropyconditionIK,eq:entropyconditionJK,eq:entropyconditionIJK}.
\end{proof}
\begin{remark}
\Cref{prop:capacityregion} and its proof readily generalizes to any number of factors. In detail, if $\Lambda\subseteq I_1\times\dots\times I_k$, then the capacity region is the union of the sets $\setbuild{(a_1,\dots,a_k)\in\nonnegativereals^k}{\forall J\subseteq[k]:\sum_{j\in J}a_j\le\entropy(J)_P}$ over all probability distributions $P\in\distributions(\Lambda)$.

For instance, when $k=2$ and $\Lambda=\{(0,0),(0,1),(1,0)\}$, then the capacity region is the convex hull of the graph of $h(p)=-p\log p-(1-p)\log(1-p)$ and that of its inverse. The sum rate is maximal along the line segment joining $(2/3,h(2/3))$ and $(h(2/3),2/3)$, the common tangent of the two curves.
\end{remark}

\begin{theorem}\label{thm:srankexponentbound}
If $P\in\distributions(\Lambda)$ and the numbers $a,b,c\ge 0$ satisfy the inequalities \labelcref{eq:entropyconditionI,eq:entropyconditionJ,eq:entropyconditionK,eq:entropyconditionIJ,eq:entropyconditionIK,eq:entropyconditionJK,eq:entropyconditionIJK},
then $\omega_s(a,b,c)\le\log\asymptoticsupportrank(\partialmmsupport{\Lambda})\le\log\bordersupportrank(\partialmmsupport{\Lambda})\le\log\supportrank(\partialmmsupport{\Lambda})$.

Consequently,
\begin{equation}\label{eq:squaremmsrankexponentbound}
\omega_s\le 3\frac{\log\supportrank(\partialmmsupport{\Lambda})}{\log\lvert\Lambda\rvert}.
\end{equation}
\end{theorem}
\begin{proof}
If $(a,b,c)$ is an achievable rate triple, then by definition there exist maps $f_n:I^n\to[\lfloor 2^{an}\rfloor]$, $g_n:I^n\to[\lfloor 2^{bn}\rfloor]$, $h_n:I^n\to[\lfloor 2^{cn}\rfloor]$, such that $(f_n\times g_n\times h_n)(\Lambda^n)=[\lfloor 2^{an}\rfloor]\times[\lfloor 2^{bn}\rfloor]\times[\lfloor 2^{cn}\rfloor]$. By \cref{cor:partialmmproductmapsrank}, this implies that $\supportrank(\partialmmsupport{\Lambda^n})\ge\supportrank(\partialmmsupport{[\lfloor 2^{an}\rfloor]\times[\lfloor 2^{bn}\rfloor]\times[\lfloor 2^{cn}\rfloor]})=\supportrank(\support(\mmtensor{\lfloor 2^{an}\rfloor}{\lfloor 2^{bn}\rfloor}{\lfloor 2^{cn}\rfloor}))$, i.e.,
\begin{equation}
\begin{split}
\omega_s(a,b,c)
 & = \lim_{n\to\infty}\frac{1}{n}\log\supportrank(\support(\mmtensor{\lfloor 2^{an}\rfloor}{\lfloor 2^{bn}\rfloor}{\lfloor 2^{cn}\rfloor}))  \\
 & \le \lim_{n\to\infty}\frac{1}{n}\log\supportrank(\partialmmsupport{\Lambda^n})  \\
 & = \log\asymptoticsupportrank(\partialmmsupport{\Lambda}).
\end{split}
\end{equation}
The inequalitites between the asymptotic support rank, the border support rank, and the support rank hold generally for any support.

As noted earlier, by joint convexity and homogeneity, $\omega_s=\omega_s(1,1,1)\le\frac{3}{a+b+c}\omega_s(a,b,c)$ holds for every $a,b,c$ with $a+b+c\neq 0$. For any distribution $P\in\distributions(\Lambda)$, the point $(a,b,c)=(\entropy(I)_P,\entropy(J|I)_P,\entropy(K|IJ)_P)$ is feasible and $\entropy(I)_P+\entropy(J|I)_P+\entropy(K|IJ)_P=\entropy(IJK)_P=\entropy(P)$ by the chain rule. Using the first part,
\begin{equation}
\omega_s
 \le \frac{3}{\entropy(P)}\omega_s(a,b,c)
 \le \frac{3}{\entropy(P)}\log\supportrank(\partialmmsupport{\Lambda}).
\end{equation}
The upper bound is minimal when $P$ is the uniform distribution on $\Lambda$, which gives \eqref{eq:squaremmsrankexponentbound}.
\end{proof}

\begin{example}
Let $\Lambda=\{(1,1,2),(1,2,1),(2,1,1),(2,2,1),(2,1,2),(1,2,2)\}\subseteq[2]\times[2]\times[2]$ ($\lvert\Lambda\rvert=6$). The border decomposition
\begin{equation}
\begin{split}
\epsilon^3\partialmmtensor{\Lambda}+O(\epsilon^4)
 & = x_{1,2}y_{1,2}z_{1,2}  \\
 &\qquad -(x_{1,2}+\epsilon^2x_{2,1})(y_{1,2}+\epsilon^2y_{2,1})(z_{1,2}+\epsilon^2z_{2,1})  \\
 &\qquad +\epsilon^2x_{2,1}(y_{1,2}+\epsilon y_{1,1})(z_{1,2}+\epsilon z_{2,2})  \\
 &\qquad +\epsilon^2(x_{1,2}+\epsilon x_{2,2})y_{2,1}(z_{1,2}+\epsilon z_{1,1})  \\
 &\qquad +\epsilon^2(x_{1,2}+\epsilon x_{1,1})(y_{1,2}+\epsilon y_{2,2})z_{2,1}
\end{split}
\end{equation}
shows that $\borderrank(\partialmmtensor{\Lambda})\le 5$, leading to the same parameters as the partial matrix multiplication tensor due to Bini, Capovani, Romani, and Lotti \cite{bini1979complexity}, with $\Lambda_{\textup{BCRL}}=\{(1,1,1),(1,1,2),(1,2,1),(1,2,2),(2,1,1),(2,1,2)\}$. Accordingly, the bound on $\omega_s$ is
\begin{equation}
\omega_s\le 3\frac{\log 5}{\log 6}=2.6947\ldots,
\end{equation}
the same as the bound on $\omega$ obtained from the tensor given in \cite{bini1979complexity} using the method of \cite{schonhage1981partial}.

However, the capacity regions are different. For instance, $C(\Lambda)$ contains every permutation of $(1,1,1/2)$, but $C(\Lambda_{\textup{BCRL}})$ does not contain $(1,1,1/2)$, since the support of the marginal on the first two factors is a subset of $\{(1,1),(1,2),(2,1)\}$, therefore the entropy is at most $\log 3<1+1$.
\end{example}

As one would expect, \cref{thm:srankexponentbound} can be combined with other methods developed previously for total matrix multiplication. We now give a version of the asymptotic sum inequality based on the proof in Bl\"aser's survey \cite{blaser2013fast}.
\begin{proposition}
Consider $p$ patterns $\Lambda_1\subseteq I_1\times J_1\times K_1$, $\Lambda_2\subseteq I_2\times J_2\times K_2$, \dots, $\Lambda_p\subseteq I_p\times J_p\times K_p$, and let $Q\in\distributions([p])$. If
\begin{equation}
(a,b,c)\in\sum_{i=1}^p Q(i)C(\Lambda_i),
\end{equation}
then
\begin{equation}\label{eq:rectangularasymptoticsuminequality}
\omega_s(a,b,c)\le\log\asymptoticsupportrank(\partialmmtensor{\Lambda_1}\oplus\dots\oplus\partialmmtensor{\Lambda_p})-\entropy(Q).
\end{equation}

Consequently,
\begin{equation}\label{eq:squareasymptoticsuminequality}
\sum_{i=1}^p\lvert\Lambda_i\rvert^{\omega_s/3}\le\asymptoticsupportrank(\partialmmtensor{\Lambda_1}\oplus\dots\oplus\partialmmtensor{\Lambda_p}).
\end{equation}
\end{proposition}
\begin{proof}
By continuity, it suffices to consider $Q$ with rational entries, i.e., $Q\in\distributions[n]([p])$ for some $n\in\positiveintegers$. Then $(\partialmmtensor{\Lambda_1}\oplus\dots\oplus\partialmmtensor{\Lambda_p})^{\otimes n}$ restricts to the tensor
\begin{equation}
\unittensor{\binom{n}{nQ}}\otimes\bigotimes_{i=1}^p\partialmmtensor{\Lambda_i}^{\otimes nQ(i)},
\end{equation}
therefore
\begin{equation}
\asymptoticsupportrank\left(\partialmmtensor{\prod_{i=1}^p\Lambda_i^{nQ(i)}}\right)
 = \asymptoticsupportrank\left(\bigotimes_{i=1}^p\partialmmtensor{\Lambda_i}^{\otimes nQ(i)}\right)
 \le \frac{\asymptoticsupportrank\left(\partialmmtensor{\Lambda_1}\oplus\dots\oplus\partialmmtensor{\Lambda_p}\right)^n}{\binom{n}{nQ}}
\end{equation}
(see \cite[Lemma 7.7]{blaser2013fast} for the inequality). Since $(na,nb,nc)\in\sum_{i=1}^p nQ(i)C(\Lambda_i)\subseteq C(\prod_{i=1}^p\Lambda_i^{nQ(i)})$,
\begin{equation}
\omega_s(a,b,c)\le\frac{1}{n}\log\asymptoticsupportrank\left(\partialmmtensor{\Lambda_1}\oplus\dots\oplus\partialmmtensor{\Lambda_p}\right)^n-\frac{1}{n}\log\binom{n}{nQ}.
\end{equation}
Finally, take the limit $n\to\infty$ (along a subsequence ensuring that $nQ$ has integer entries) to obtain \eqref{eq:rectangularasymptoticsuminequality}.

To prove \eqref{eq:squareasymptoticsuminequality}, we continue as
\begin{equation}
\begin{split}
\frac{a+b+c}{3}\omega_s
 & \le \omega_s(a,b,c)  \\
 & \le \log\asymptoticsupportrank(\partialmmtensor{\Lambda_1}\oplus\dots\oplus\partialmmtensor{\Lambda_p})-\entropy(Q),
\end{split}
\end{equation}
and use that
\begin{equation}
\max\setbuild{a+b+c}{(a,b,c)\in\sum_{i=1}^p Q(i)C(\Lambda_i)}=\sum_{i=1}^p Q(i)\log\lvert\Lambda_i\rvert
\end{equation}
to get the best bound. Rearranging,
\begin{equation}
\entropy(Q)+\sum_{i=1}^p Q(i)\log\lvert\Lambda_i\rvert^{\omega_s/3}\le\log\asymptoticsupportrank(\partialmmtensor{\Lambda_1}\oplus\dots\oplus\partialmmtensor{\Lambda_p}).
\end{equation}
We choose (optimally) $Q(i)=\frac{\lvert\Lambda_i\rvert^{\omega_s/3}}{\sum_{j=1}^p\lvert\Lambda_j\rvert^{\omega_s/3}}$, which results in \eqref{eq:squareasymptoticsuminequality}.
\end{proof}

Similarly, the laser method can be extended to block-tight tensors that have partial matrix multiplication tensors as blocks. For completeness we state the result and refer to \cite[Theorem 9.10]{blaser2013fast} or the papers \cite{coppersmith1990matrix,strassen1991degeneration} for a proof (see also \cite[Section 15.7]{burgisser1997algebraic}), which generalises in a straightforward way.
\begin{proposition}
Let $t$ have a block decomposition $\mathcal{D}$ such that
\begin{enumerate}
\item $\support_\mathcal{D}t\subseteq I\times J\times K$ is tight
\item for all $(i,j,k)\in\support_\mathcal{D}t$, $t_{i,j,k}$ is a partial matrix multiplication tensor with pattern $\Lambda_{i,j,k}$.
\end{enumerate}
Let $Q$ be an arbitrary probability distribution on $\support_\mathcal{D}t$ and
\begin{equation}
(a,b,c)\in\sum_{(i,j,k)\in\support_\mathcal{D}t} Q(i,j,k)C(\Lambda_{i,j,k}).
\end{equation}
Then
\begin{equation}
\omega_s(a,b,c)+\min\{\entropy(Q_I),\entropy(Q_J),\entropy(Q_K)\}+\entropy(Q)-\max_P\entropy(P)\le\asymptoticrank(t),
\end{equation}
where the maximum is over probability distributions $P\in\distributions(\support_\mathcal{D}t)$ such that $P_I=Q_I$, $P_J=Q_J$, and $P_K=Q_K$.

In particular,
\begin{equation}
\min\{\entropy(Q_I),\entropy(Q_J),\entropy(Q_K)\}+\entropy(Q)-\max_P\entropy(P)+\frac{\omega_s}{3}\sum_{(i,j,k)\in\support_\mathcal{D}t} Q(i,j,k)\log\lvert\Lambda_{i,j,k}\rvert\le\asymptoticrank(t),
\end{equation}
\end{proposition}

\section*{Acknowledgement}

The project has been implemented with the support provided by the Ministry of Culture and Innovation of Hungary from the National Research, Development and Innovation Fund, financed under the FK~146643 funding scheme. Supported by the J\'anos Bolyai Research Scholarship of the Hungarian Academy of Sciences.

\bibliography{refs}{}

\end{document}